\documentclass{amsart}

\usepackage{amssymb,amsmath,tikz,bbm}
\usepackage[bookmarks=true,%
    colorlinks=true,%
    linkcolor=blue,%
    citecolor=blue,%
    filecolor=blue,%
    menucolor=blue,%
    urlcolor=blue,%
    breaklinks=true]{hyperref}
\newcommand{\Z}{\mathbb{Z}}
\newcommand{\C}{\mathbb{C}}
\newcommand{\R}{\mathbb{R}}

\newcommand{\T}{\mathbb{T}}
\newcommand{\imi}{\mathsf{i}}
\newcommand{\gauss}[1]{\langle #1\rangle}
\newcommand{\fourier}[2]{\langle #1;#2\rangle}

\newcommand{\fad}[2]{\Phi_{#1}\left(#2\right)}
\newcommand{\fvm}[2]{W_\text{FV}\left(#2,#1\right)}
\newcommand{\gfvm}[2]{\operatorname{w}\!\left(#2,#1\right)}
\newcommand{\mfun}{\operatorname{M}}
\newcommand{\gri}{\operatorname{\mathfrak{p}}}
\newcommand{\wgfvm}[2]{\operatorname{\check w}\!\left(#2,#1\right)}
\newcommand{\ffvm}[2]{\operatorname{\tilde w}\!\left(#2,#1\right)}
\newcommand{\poc}[2]{\left(#1\right)_{#2}}

\newtheorem{proposition}{Proposition}

\theoremstyle{remark}
\newtheorem{definition}{Definition}

\title{The Yang--Baxter relation and gauge invariance}
\author{Rinat Kashaev}
\address{Section de math\'ematiques, Universit\'e de Gen\`eve,
2-4 rue du Li\`evre, 1211 Gen\`eve 4, Suisse\\}
\email{Rinat.Kashaev@unige.ch}
\date{October 11, 2015}
\thanks{Supported in part by Swiss National Science Foundation}
 \dedicatory{Dedicated to Professor Rodney Baxter on the occasion
     of his $75$th birthday}
\begin{document}

\begin{abstract} 
Starting from a quantum dilogarithm over a Pontryagin self-dual LCA group $A$, we construct an operator solution of the Yang--Baxter equation generalizing the solution of the Faddeev--Volkov model. Based on a specific choice of a subgroup $B\subset A$ and by using the Weil transformation, we also give a new non-operator interpretation of the Yang--Baxter relation. That allows us to construct a lattice QFT-model of IRF-type with gauge invariance under independent $B$-translations of local `spin' variables. 
\end{abstract}
\maketitle
\section{Introduction}
The recent progress in the construction of new  integrable models of lattice statistical physics and quantum field theory (QFT) is related with models where local `spins' take infinitely many values, and the Boltzmann weights of local interactions are given in terms of highly transcendental special functions, see e.g. \cite{MR1846786,MR2353414,MR2396233,MR3019403,MR3138832,MR2868616,MR2730793,MR3246845,MR3190461,MR2861174}. Some of those functions, particularly the one known under the name `quantum dilogarithm' \cite{MR1264393}, also find applications in quantum topology, see e.g. \cite{MR1607296,MR1865275,MR3262519,KashaevMarino2015}. 

The notion of a quantum dilogarithm associated with an arbitrary Pontryagin self-dual locally compact Abelian (LCA) group has been introduced in \cite{AndersenKashaev2014} in the context of quantum topology. In this paper, we discuss this notion from the perspective of lattice integrable systems and the Yang--Baxter equation.
The prototypical example for the definition of \cite{AndersenKashaev2014} is Faddeev's quantum dilogarithm $\fad{\mathsf{b}}{x}$ which is a meromorphic function on $\C$ defined by the formula
\begin{equation}\label{qdl}
\fad{\mathsf{b}}{x}
={\left(-q_\mathsf{b}e^{2 \pi \mathsf{b} x};q^2_\mathsf{b}\right)_\infty\over
\left(-\bar q_\mathsf{b} e^{2 \pi \mathsf{b}^{-1} x};\bar q^2_\mathsf{b}\right)_\infty},\quad q_\mathsf{b}\equiv e^{ \pi \imi \mathsf{b}^2}, \quad 
\bar q_\mathsf{b}\equiv e^{- \pi \imi \mathsf{b}^{-2}}, \quad 
\operatorname{Im}(\mathsf{b}^2) >0. 
\end{equation}
There are at least two integral representations, one given by Faddeev \cite{MR1345554}
\begin{equation}
\fad{\mathsf{b}}{x}=\exp\left( \int_{\mathbb{R}+\imi\epsilon}
\frac{e^{-2\imi xz}}{4\sinh(z\mathsf{b})\sinh(z\mathsf{b}^{-1})}
\frac{\operatorname{d}\!z}{z} \right)
\end{equation}
and another by Woronowicz \cite{MR1770545}
\begin{equation}
\fad{\mathsf{b}}{x}=\exp\left( {\imi\over 2\pi}\int_{\R}\log\left(e^{\mathsf{b}^{2}t+2\pi \mathsf{b}x}+1\right){\operatorname{d}\!t\over e^{t}+1}\right),
\end{equation}
which allow to analytically continue the function to real values of $\mathsf{b}$. 
Properties of $\fad{\mathsf{b}}{x}$
include  the inversion relation
\begin{equation}
\fad{\mathsf{b}}{x}\fad{\mathsf{b}}{-x}=\fad{\mathsf{b}}{0}^2e^{\pi\imi x^2}
\end{equation}
and the operator five term relation~\cite{Faddeev1994}
\begin{equation}\label{5term}
\fad{\mathsf{b}}{\mathbf{p}}\fad{\mathsf{b}}{\mathbf{q}}=\fad{\mathsf{b}}{\mathbf{q}}\fad{\mathsf{b}}{\mathbf{p}+\mathbf{q}}\fad{\mathsf{b}}{\mathbf{p}}
\end{equation}
where $\mathbf{p}$ and $\mathbf{q}$ are self-adjoint Heisenberg operators acting in the Hilbert space $L^2(\R)$ as differentiation and  multiplication operators respectively,
\begin{equation}
\langle x|\mathbf{p}={1\over 2\pi\imi}{\partial\over\partial x}\langle x|,\quad\langle x|\mathbf{q}=x\langle x|,
\end{equation}
where we use Dirac's bra-ket notation.
In the case where the  condition 
\begin{equation}
(1-|\mathsf{b}|)\operatorname{Im}(\mathsf{b})=0
\end{equation}
is satisfied, one has the unitarity property
\begin{equation}
|\fad{\mathsf{b}}{x}|=1,\quad \forall x\in\R.
\end{equation}

The generalization of Faddeev's quantum dilogarithm to the general setting of arbitrary self-dual  LCA groups  can be obtained by interpreting appropriately the Heisenberg operators in \eqref{5term}, for example, the way it was done in \cite{AndersenKashaev2014}. Equivalently, one can achieve that generalization by rewriting \eqref{5term} in a form not containing the differentiation operator $\mathbf{p}$. To this end, we define a unitary Fourier transformation operator $\mathbf{F}$ by the integral kernel (in Dirac's bra-ket notation)
\begin{equation}
\langle x|\mathbf{F}| y\rangle=e^{2\pi\imi xy}.
\end{equation}
It is easily verified the operator  equalities
\begin{equation}
\mathbf{F}\mathbf{q}=\mathbf{p}\mathbf{F},\quad e^{\pi\imi \mathbf{q}^2}(\mathbf{p}+\mathbf{q})=\mathbf{p}e^{\pi\imi \mathbf{q}^2}
\end{equation}
which imply that 
\begin{equation}
g(\mathbf{p})=\mathbf{F}g(\mathbf{q})\mathbf{F}^*,\quad g(\mathbf{p}+\mathbf{q})=e^{-\pi\imi \mathbf{q}^2}g(\mathbf{p})e^{\pi\imi \mathbf{q}^2}
\end{equation}
for any function $g\in \C^\R$. These relations  can be taken as definitions of their left hand sides in any Hilbert space as soon as the operators $\mathbf{F}$ and $g(\mathbf{q})$ are defined. We follow this line of reasoning in Section~\ref{sec2} where the Hilbert space is $L^2(A)$ with $A$ being an arbitrary self-dual LCA group.

\subsection{The Faddeev--Volkov weight function}
The Faddeev--Volkov model, as defined and studied in \cite{MR2353414,MR2396233}, is characterized by a weight function of the form
\begin{equation}
\fvm{y}{x}={\fad{\mathsf{b}}{x-y}\over \fad{\mathsf{b}}{x+y}}e^{2\pi\imi xy}
\end{equation}
which satisfies the operator Yang--Baxter relation
\begin{equation}
\fvm{x}{\mathbf{p}}\fvm{x+y}{\mathbf{q}}\fvm{y}{\mathbf{p}}=\fvm{y}{\mathbf{q}}\fvm{x+y}{\mathbf{p}}\fvm{x}{\mathbf{q}}.
\end{equation}
In this paper, we generalize the Faddeev--Volkov model to the setting of arbitrary self-dual LCA groups by using the generalized quantum dilogarithm function and give a novel way of interpreting the corresponding Yang--Baxter relation in terms of a lattice IRF-model of quantum field theory with gauge symmetry.

The paper is organized as follows. In Section~\ref{sec2}, we recall the definition of a quantum dilogarithm over a LCA group and prove the operator Yang--Baxter relation for the generalized Faddeev--Volkov weight function. In Section~\ref{sec3}, we discuss non-operator forms of the Yang--Baxter relation. Namely, we start by recalling the star-triangle relation and then proceed to the main result of the paper: construction of an IRF weight function having quasi-invariance property. Section~\ref{sec4} is devoted to examples of realization of the general  construction.

\section{Generalization to arbitrary self-dual LCA groups}\label{sec2}
Let $\T$ denote the complex circle group. For any LCA group $A$, we denote by  $\hat A$ its Pontryagin dual which is also a LCA group.  When $A$  is self-dual,  we assume that there exists an isomorphism $f\colon A\to \hat A$ such that there exists a function $\gauss{\cdot}\colon A\to \T$, called \emph{Gaussian exponential}, which is symmetric $\gauss{x}=\gauss{-x}$ and trivializes the group 2-cocycle $f(x)(y)$, i.e.
\begin{equation}
\fourier{x}{y}\equiv f(x)(y)={\gauss{x+y}\over\gauss{x}\gauss{y}}.
\end{equation}
We fix the normalized Haar measure $\operatorname{d}\!x$ on $A$ by the condition
\begin{equation}
\int_{A^2}\fourier{x}{y}\operatorname{d}\!x\operatorname{d}\!y=1
\end{equation}
and define a unitary operator $\mathbf{F}$ in the Hilbert space $L^2(A)$ corresponding to the Fourier transformation with the integral kernel (in Dirac's bra-ket notation)
\begin{equation}
\langle x|\mathbf{F}| y\rangle=\fourier{x}{y}.
\end{equation}
That means that
\begin{equation}
\langle x|\mathbf{F}f\rangle\equiv (\mathbf{F}f)(x)=\int_A \fourier{x}{y}f(y)\operatorname{d}\!y,\quad\forall f\in L^2(A).
\end{equation}
Additionally, to any function $g\in \C^A$, we associate operators $g(\mathbf{q})$, $g(\mathbf{p})$ and $g(\mathbf{p}+\mathbf{q})$, where the first one  acts diagonally by multiplication by $g$, i.e.
\begin{equation}\label{dac}
\langle x |g(\mathbf{q})f\rangle\equiv (g(\mathbf{q})f)(x)=g(x)f(x),\quad\forall f\in L^2(A),
\end{equation}
while the others are given by certain conjugations of the first one:
\begin{equation}\label{p.p+q}
g(\mathbf{p})\equiv\mathbf{F}g(\mathbf{q})\mathbf{F}^*,\quad g(\mathbf{p}+\mathbf{q})\equiv\gauss{\mathbf{q}}^*g(\mathbf{p})\gauss{\mathbf{q}}.
\end{equation}

\begin{definition}
 A \emph{quantum dilogarithm over a self-dual LCA group $A$} is a function $\phi\colon A\to \T$ which satisfies the relation
 \begin{equation}
\phi(x)\phi(-x)=\phi(0)^2\gauss{x},\quad \forall x\in A,
\end{equation}
and the operator relation in the Hilbert space $L^2(A)$ 
\begin{equation}\label{gpr}
\phi(\mathbf{p})\phi(\mathbf{q})=\phi(\mathbf{q})\phi(\mathbf{p}+\mathbf{q})\phi(\mathbf{p}).
\end{equation}
\end{definition}
\begin{proposition}
 The operator relation~\eqref{gpr} is equivalent to the integral relation
 \begin{equation}\label{ipent}
\phi(x)\phi(y)=\int_{A^3}Q_{x,y}^{u,v,w}\phi(w)\phi(v)\phi(u)\operatorname{d}(u,v,w)
\end{equation}
where
\begin{equation}
Q_{x,y}^{u,v,w}\equiv \gamma{\fourier{u-x}{w-y}\over\gauss{u-v+w}},\quad \gamma\equiv \int_A\gauss{z}\operatorname{d}\!z.
\end{equation}
\end{proposition}
\begin{proof}
We multiply from left the both sides of \eqref{gpr} by $\mathbf{F}^*$ and calculate separately the integral kernels of the right and left hand sides. 
 By using definitions~\eqref{dac} and \eqref{p.p+q}, we have 
\begin{equation}\label{lhs1}
 \langle x|\mathbf{F}^*\phi(\mathbf{p})\phi(\mathbf{q})|y\rangle=\phi(x) \langle x|\mathbf{F}^*|y\rangle\phi(y)={\phi(x)\phi(y)\over\fourier{x}{y}}
 \end{equation}
 and
\begin{multline}\label{rhs1}
\langle x|\mathbf{F}^*\phi(\mathbf{q})\phi(\mathbf{p}+\mathbf{q})\phi(\mathbf{p})|y\rangle=\langle x|\mathbf{F}^*\phi(\mathbf{q})\gauss{\mathbf{q}}^*\mathbf{F}\phi(\mathbf{q})\mathbf{F}^*\gauss{\mathbf{q}}\mathbf{F}\phi(\mathbf{q})\mathbf{F}^*|y\rangle\\
=\int_{A^4}\langle x|\mathbf{F}^*|w\rangle\phi(w)\overline{\gauss{w}}\langle w|\mathbf{F}|v\rangle\phi(v)\langle v|\mathbf{F}^*|z\rangle\gauss{z}\langle z|\mathbf{F}|u\rangle\phi(u)\langle u|\mathbf{F}^*|y\rangle\operatorname{d}\!u\operatorname{d}\!v\operatorname{d}\!w\operatorname{d}\!z\\
=\int_{A^4}{\fourier{v-x}{w}\over\fourier{u}{y}\gauss{w}}\phi(w)\phi(v)\phi(u)\fourier{u-v}{z}\gauss{z}\operatorname{d}\!u\operatorname{d}\!v\operatorname{d}\!w\operatorname{d}\!z\\
=\gamma\int_{A^3}{\fourier{v-x}{w}\over\fourier{u}{y}\gauss{w}\gauss{u-v}}\phi(w)\phi(v)\phi(u)\operatorname{d}\!u\operatorname{d}\!v\operatorname{d}\!w\\
=\gamma\int_{A^3}{\fourier{u-x}{w-y}\over\fourier{x}{y}\gauss{u-v+w}}\phi(w)\phi(v)\phi(u)\operatorname{d}\!u\operatorname{d}\!v\operatorname{d}\!w.
\end{multline}
Equating \eqref{lhs1} to \eqref{rhs1}, we obtain the integral relation~\eqref{ipent}.
\end{proof}

\subsection{The generalized Faddeev--Volkov weight function}
Following \cite{MR2353414,MR2396233}, we define the generalized Faddeev--Volkov weight function by the formula
 \begin{equation}\label{gfvwf}
\gfvm{y}{x}\equiv{\phi(x-y)\over\phi(x+y)}\fourier{x}{y}={\phi(x-y)\phi(-x-y)\over\phi(0)^{2}\gauss{x}\gauss{y}},
\end{equation}
where the second equality reveals the symmetry property
\begin{equation}
\gfvm{y}{x}=\gfvm{y}{-x}.
\end{equation}
\begin{proposition}
The generalized Faddeev--Volkov weight function \eqref{gfvwf} satisfies the following operator Yang--Baxter relation
\begin{equation}\label{gybe}
\gfvm{x}{\mathbf{p}}\gfvm{x+y}{\mathbf{q}}\gfvm{y}{\mathbf{p}}=\gfvm{y}{\mathbf{q}}\gfvm{x+y}{\mathbf{p}}\gfvm{x}{\mathbf{q}}.
\end{equation}
\end{proposition}
\begin{proof}
We start by transforming the left hand side of \eqref{gybe}:
\begin{multline}
 \gfvm{x}{\mathbf{p}}\gfvm{x+y}{\mathbf{q}}\gfvm{y}{\mathbf{p}}\\={\phi(\mathbf{p}-x)\over\phi(\mathbf{p}+x)}\fourier{\mathbf{p}}{x}{\phi(\mathbf{q}-x-y)\over\phi(\mathbf{q}+x+y)}\fourier{\mathbf{q}}{x+y}{\phi(\mathbf{p}-y)\over\phi(\mathbf{p}+y)}\fourier{\mathbf{p}}{y}\\
 ={\phi(\mathbf{p}-x)\over\phi(\mathbf{p}+x)}{\phi(\mathbf{q}-y)\over\phi(\mathbf{q}+2x+y)}{\phi(\mathbf{p}-x-2y)\over\phi(\mathbf{p}-x)}\fourier{\mathbf{p}}{x}\fourier{\mathbf{q}}{x+y}\fourier{\mathbf{p}}{y}\\
 ={\phi(\mathbf{p}-x)\over\phi(\mathbf{p}+x)}{\phi(\mathbf{q}-y)\over\phi(\mathbf{q}+2x+y)}{\phi(\mathbf{p}-x-2y)\over\phi(\mathbf{p}-x)}\fourier{x}{x+y}\fourier{\mathbf{q}}{x+y}\fourier{\mathbf{p}}{x+y},
\end{multline}
and, similarly, the left hand side of \eqref{gybe}:
\begin{multline}
 \gfvm{y}{\mathbf{q}}\gfvm{x+y}{\mathbf{p}}\gfvm{x}{\mathbf{q}}\\={\phi(\mathbf{q}-y)\over\phi(\mathbf{q}+y)}\fourier{\mathbf{q}}{y}{\phi(\mathbf{p}-x-y)\over\phi(\mathbf{p}+x+y)}\fourier{\mathbf{p}}{x+y}{\phi(\mathbf{q}-x)\over\phi(\mathbf{q}+x)}\fourier{\mathbf{q}}{x}\\
 ={\phi(\mathbf{q}-y)\over\phi(\mathbf{q}+y)}{\phi(\mathbf{p}-x-2y)\over\phi(\mathbf{p}+x)}{\phi(\mathbf{q}+y)\over\phi(\mathbf{q}+2x+y)}\fourier{\mathbf{q}}{y}\fourier{\mathbf{p}}{x+y}\fourier{\mathbf{q}}{x}\\
={\phi(\mathbf{q}-y)\over\phi(\mathbf{q}+y)}{\phi(\mathbf{p}-x-2y)\over\phi(\mathbf{p}+x)}{\phi(\mathbf{q}+y)\over\phi(\mathbf{q}+2x+y)}\fourier{x}{x+y}\fourier{\mathbf{q}}{x+y}\fourier{\mathbf{p}}{x+y}.
\end{multline}
Comparing the obtained expressions, we conclude that \eqref{gybe} is equivalent to the equality
\begin{multline}
\phi^*(\mathbf{p}+x) \phi(\mathbf{p}-x)\phi(\mathbf{q}-y)\phi^*(\mathbf{q}+2x+y)\phi^*(\mathbf{p}-x)\phi(\mathbf{p}-x-2y)\\=
 \phi(\mathbf{q}-y)\phi^*(\mathbf{q}+y)\phi^*(\mathbf{p}+x)\phi(\mathbf{p}-x-2y)\phi(\mathbf{q}+y)\phi^*(\mathbf{q}+2x+y)
\end{multline}
where we have chosen specific orderings for the commuting terms. These orderings allow us to apply the five term relation to the second and third terms as well as to the forth and fifth terms simultaneously in both sides of the equality:
\begin{multline}
\phi^*(\mathbf{p}+x)\phi(\mathbf{q}-y)\phi(\mathbf{p}+\mathbf{q}-x-y) \phi(\mathbf{p}-x)\\
\times\phi^*(\mathbf{p}-x)\phi^*(\mathbf{p}+\mathbf{q}+x+y)\phi^*(\mathbf{q}+2x+y)\phi(\mathbf{p}-x-2y)\\=
 \phi(\mathbf{q}-y)\phi^*(\mathbf{p}+x)\phi^*(\mathbf{p}+\mathbf{q}+x+y)\phi^*(\mathbf{q}+y)\\
 \times\phi(\mathbf{q}+y)\phi(\mathbf{p}+\mathbf{q}-x-y)\phi(\mathbf{p}-x-2y)\phi^*(\mathbf{q}+2x+y)
\end{multline}
and the forth and fifth terms  cancel each other in both sides:
\begin{multline}
\phi^*(\mathbf{p}+x)\phi(\mathbf{q}-y){\phi(\mathbf{p}+\mathbf{q}-x-y)\over \phi(\mathbf{p}+\mathbf{q}+x+y)}\phi^*(\mathbf{q}+2x+y)\phi(\mathbf{p}-x-2y)\\=
 \phi(\mathbf{q}-y)\phi^*(\mathbf{p}+x){\phi(\mathbf{p}+\mathbf{q}-x-y)\over \phi(\mathbf{p}+\mathbf{q}+x+y)}\phi(\mathbf{p}-x-2y)\phi^*(\mathbf{q}+2x+y).
\end{multline}
To proceed further, we multiply the both sides by the inverse of the first term of the left hand side from the left and by the inverse of the last term of the right hand side from the right:
\begin{multline}
\phi(\mathbf{q}-y){\phi(\mathbf{p}+\mathbf{q}-x-y)\over \phi(\mathbf{p}+\mathbf{q}+x+y)}\phi^*(\mathbf{q}+2x+y)\phi(\mathbf{p}-x-2y)\phi(\mathbf{q}+2x+y)\\=
\phi(\mathbf{p}+x) \phi(\mathbf{q}-y)\phi^*(\mathbf{p}+x){\phi(\mathbf{p}+\mathbf{q}-x-y)\over \phi(\mathbf{p}+\mathbf{q}+x+y)}\phi(\mathbf{p}-x-2y)
\end{multline}
where we can apply again the five term relation to the last two terms of the left hand side and the first two terms of the right hand side:
\begin{multline}
\phi(\mathbf{q}-y){\phi(\mathbf{p}+\mathbf{q}-x-y)\over \phi(\mathbf{p}+\mathbf{q}+x+y)}\phi(\mathbf{p}+\mathbf{q}+x-y)\phi(\mathbf{p}-x-2y)\\=
 \phi(\mathbf{q}-y) \phi(\mathbf{p}+\mathbf{q}+x-y){\phi(\mathbf{p}+\mathbf{q}-x-y)\over \phi(\mathbf{p}+\mathbf{q}+x+y)}\phi(\mathbf{p}-x-2y).
\end{multline}
The obtained equality is trivially true because the middle terms are commuting.
\end{proof}
\section{Non-operator forms of the Yang--Baxter relation}\label{sec3}
The standard non-operator form of the Yang--Baxter relation~\eqref{gybe} is the \emph{star-triangle relation}. It can be obtained as follows.

We start by removing the operator $\mathbf{p}$ in \eqref{gybe}
\begin{equation}\label{cybe}
\mathbf{F}\gfvm{x}{\mathbf{q}}\mathbf{F}^*\gfvm{x+y}{\mathbf{q}}\mathbf{F}\gfvm{y}{\mathbf{q}}=\gfvm{y}{\mathbf{q}}\mathbf{F}\gfvm{x+y}{\mathbf{q}}\mathbf{F}^*\gfvm{x}{\mathbf{q}}\mathbf{F}
\end{equation}
and then equate the operator kernels
\begin{equation}
\langle u|\mathbf{F}\gfvm{x}{\mathbf{q}}\mathbf{F}^*\gfvm{x+y}{\mathbf{q}}\mathbf{F}|v\rangle\gfvm{y}{v}=\gfvm{y}{u}\langle u|\mathbf{F}\gfvm{x+y}{\mathbf{q}}\mathbf{F}^*\gfvm{x}{\mathbf{q}}\mathbf{F}|v\rangle.
\end{equation}
Inserting two decompositions of unity in each side, 
\begin{multline}
\int_{A^2}\langle u|\mathbf{F}|s\rangle\langle s|\mathbf{F}^*|t\rangle\langle t|\mathbf{F}|v\rangle\gfvm{x}{s}\gfvm{x+y}{t}\gfvm{y}{v}
\operatorname{d}(s,t)\\
=\int_{A^2}\gfvm{y}{u}\gfvm{x+y}{s}\gfvm{x}{t}\langle u|\mathbf{F}|s\rangle\langle s|\mathbf{F}^*|t\rangle\langle t|\mathbf{F}|v\rangle\operatorname{d}(s,t)
\end{multline}
we obtain an integral identity
\begin{multline}\label{int-id}
\int_{A^2}{\fourier{u}{s}\fourier{t}{v}\over\fourier{s}{t} }\gfvm{x}{s}\gfvm{x+y}{t}\gfvm{y}{v}
\operatorname{d}(s,t)\\
=\int_{A^2}{\fourier{u}{s}\fourier{t}{v}\over\fourier{s}{t} }\gfvm{y}{u}\gfvm{x+y}{s}\gfvm{x}{t}\operatorname{d}(s,t).
\end{multline}
Now, we absorb the $s$-integrations by introducing  a new weight function through the Fourier transformation
\begin{equation}
\ffvm{x}{y}=\int_A\fourier{y}{z}\gfvm{x}{z}\operatorname{d}\!z
\end{equation}
\begin{multline}
\int_{A}\fourier{t}{v}\ffvm{x}{u-t}\gfvm{x+y}{t}\gfvm{y}{v}
\operatorname{d}\!t\\
=\int_{A}\fourier{t}{v}\gfvm{y}{u}\ffvm{x+y}{u-t}\gfvm{x}{t}\operatorname{d}\!t,
\end{multline}
then multiply the both sides by $\fourier{-s}{v}$ and integrate over $v$:
\begin{multline}
\int_{A^2}\fourier{t-s}{v}\ffvm{x}{u-t}\gfvm{x+y}{t}\gfvm{y}{v}
\operatorname{d}(t,v)\\
=\int_{A^2}\fourier{t-s}{v}\gfvm{y}{u}\ffvm{x+y}{u-t}\gfvm{x}{t}\operatorname{d}\!(t,v)
\end{multline}
where, in the left hand side the $v$-integrations can be absorbed into Fourier transformed weight function, while in the right hand side, it produces the delta-function $\delta(t-s)$ which allows to remove the $t$-integration:
\begin{equation}
\int_{A}\ffvm{x}{u-t}\gfvm{x+y}{t}\ffvm{y}{t-s}
\operatorname{d}\!t
=\gfvm{y}{u}\ffvm{x+y}{u-s}\gfvm{x}{s}.
\end{equation}
This is the standard form of the star-triangle relation with the local  `spin'-variables taking their values in the group $A$.
\subsection{The Weil transformation and IRF-models}
Let $B\subset A$ be a closed subgroup such that
\begin{equation}
B=B^\perp\equiv\left\{x\in A\ |\ \fourier{x}{b}=1,\ \forall b\in B\right\}.
\end{equation}
We have a group homomorphism
\begin{equation}
\gri\colon A\to \hat B,\quad \gri(x)(b)=\fourier{x}{b}
\end{equation}
which induces a natural group isomorphism
\begin{equation}
A/B\simeq\hat B.
\end{equation}

We define a tempered distribution $\wgfvm{x}{s,t}\in\mathcal{S}'(A^3)$ as the \emph{Weil transformation}~\footnote{In the context of arbitrary LCA groups, that transformation in full generality first has been introduced and used by A.~Weil  in his proof of the Plancherel formula \cite{MR0005741}. In the particular case of $A=\R$ and $B=\Z$, the transformation subsequently was used in \cite{MR0039154,Zak1967}.} of the 
generalized Faddeev--Volkov weight function
\begin{equation}
\wgfvm{x}{s,t}\equiv\int_B\gfvm{x}{s+b}\fourier{t}{b}\operatorname{d}\!b
\end{equation}
which is quasi $B$-invariant in the first argument
\begin{equation}
\wgfvm{x}{s+b,t}=\fourier{-b}{t}\wgfvm{x}{s,t},\quad \forall (b,s,t,x)\in B\times A^3
\end{equation}
and $B$-invariant in the second argument 
\begin{equation}
\wgfvm{x}{s,t+b}=\wgfvm{x}{s,t},\quad \forall (b,s,t,x)\in B\times A^3.
\end{equation}
These properties allow us to consider  $\wgfvm{x}{s,t}$ as a section of a complex line bundle over the LCA group $(A/B)^2$.
Being a specific case of the Fourier transformation, the Weil transformation is invertible
with the inverse transformation given by the integral
\begin{equation}
\gfvm{x}{s}=\int_{A/B}\wgfvm{x}{s,t}\operatorname{d}\!t.
\end{equation}
Using the Weil transformation, we can transform the integral identity~\eqref{int-id} to a form in terms of the function $\wgfvm{x}{s,t}$ as follows. 
\begin{proposition}
The integral identity~\eqref{int-id} is equivalent to the integral identity
\begin{multline}\label{int-id1}
\fourier{u}{p}\int_{A/B}\fourier{v-p}{t}\wgfvm{x}{p,u-t}\wgfvm{x+y}{t,v-p}\wgfvm{y}{v,t-q}
\operatorname{d}\!t\\
=\fourier{v}{q}\int_{A/B}\fourier{u-q}{s}\wgfvm{y}{u,s-p}\wgfvm{x+y}{s,u-q}\wgfvm{x}{q,v-s}\operatorname{d}\!s.
\end{multline}
\end{proposition}
\begin{proof}
The initial procedure of transformation consists in splitting each integration over $A$ as integration over $B$ followed by integration over $A/B$ and removing the integrations over $B$. 
By splitting the $s$-integrations in both sides of \eqref{int-id}, we have
\begin{multline}
\int_{A/B\times B\times A}\fourier{u-t}{s+b}\fourier{t}{v}\gfvm{x}{s+b}\gfvm{x+y}{t}\gfvm{y}{v}
\operatorname{d}(s,b,t)\\
=\int_{A/B\times B\times A}\fourier{u-t}{s+b}\fourier{t}{v}\gfvm{y}{u}\gfvm{x+y}{s+b}\gfvm{x}{t}\operatorname{d}(s,b,t)
\end{multline}
so that the $b$-integrations can be absorbed into the Weil transforms:
\begin{multline}
\int_{A/B\times A}\fourier{u}{s}\fourier{t}{v-s}\wgfvm{x}{s,u-t}\gfvm{x+y}{t}\gfvm{y}{v}
\operatorname{d}(s,t)\\
=\int_{A/B\times A}\fourier{u}{s}\fourier{t}{v-s}\gfvm{y}{u}\wgfvm{x+y}{s,u-t}\gfvm{x}{t}\operatorname{d}(s,t).
\end{multline}
Repeating the splitting  procedure for the  $t$-integrations, we have
\begin{multline}
\int_{A/B\times A/B\times B}\fourier{u}{s}\fourier{t+b}{v-s}\wgfvm{x}{s,u-t}\gfvm{x+y}{t+b}\gfvm{y}{v}
\operatorname{d}(s,t,b)\\
=\int_{A/B\times A/B\times B}\fourier{u}{s}\fourier{t+b}{v-s}\gfvm{y}{u}\wgfvm{x+y}{s,u-t}\gfvm{x}{t+b}\operatorname{d}(s,t,b)
\end{multline}
and again absorbing the  $b$-integrations into the Weil transformations:
\begin{multline}
\int_{A/B\times A/B}\fourier{u}{s}\fourier{t}{v-s}\wgfvm{x}{s,u-t}\wgfvm{x+y}{t,v-s}\gfvm{y}{v}
\operatorname{d}(s,t)\\
=\int_{A/B\times A/B}\fourier{u}{s}\fourier{t}{v-s}\gfvm{y}{u}\wgfvm{x+y}{s,u-t}\wgfvm{x}{t,v-s}\operatorname{d}(s,t).
\end{multline}
In the obtained integral equality, we now perform the Weil transformation over the variable $u$, 
i.e. we shift $u$ by $b$, $u\mapsto u+b$, then multiply both sides of the equality by $\fourier{b}{-p}$ and integrate over $b$:
\begin{multline}
\int_{A/B\times A/B\times B}\fourier{b}{-p}\fourier{u+b}{s}\fourier{t}{v-s}\\
\times\wgfvm{x}{s,u-t}\wgfvm{x+y}{t,v-s}\gfvm{y}{v}
\operatorname{d}(s,t,b)\\
=\int_{A/B\times A/B\times B}\fourier{b}{-p}\fourier{u+b}{s}\fourier{t}{v-s}\\
\times\gfvm{y}{u+b}\wgfvm{x+y}{s,u-t}\wgfvm{x}{t,v-s}\operatorname{d}(s,t,b).
\end{multline}
In the right hand side of this equality, the $b$-integration can be absorbed into the Weil transformation, while in the left hand side it produces the
delta-function $\delta_{A/B}(s-p)$ on the the group $A/B$ which removes the $s$-integration:
\begin{multline}
\fourier{u}{p}\int_{A/B}\fourier{t}{v-p}\wgfvm{x}{p,u-t}\wgfvm{x+y}{t,v-p}\gfvm{y}{v}
\operatorname{d}\!t\\
=\int_{A/B\times A/B}\fourier{u}{s}\fourier{t}{v-s}\wgfvm{y}{u,s-p}\wgfvm{x+y}{s,u-t}\wgfvm{x}{t,v-s}\operatorname{d}(s,t).
\end{multline}
We repeat the previous procedure for the variable $v$: shift $v$ by $b$, $u\mapsto v+b$, then multiply both sides of the equality by $\fourier{b}{-q}$ and integrate over $b$:
\begin{multline}
\fourier{u}{p}\int_{A/B\times B}\fourier{b}{-q}\fourier{t}{v-p+b}\\
\times\wgfvm{x}{p,u-t}\wgfvm{x+y}{t,v-p}\gfvm{y}{v+b}
\operatorname{d}(t,b)\\
=\int_{A/B\times A/B\times B}\fourier{b}{-q}\fourier{u}{s}\fourier{t}{v-s+b}\\
\times\wgfvm{y}{u,s-p}\wgfvm{x+y}{s,u-t}\wgfvm{x}{t,v-s}\operatorname{d}(s,t,b).
\end{multline}
This time, the $b$-integration in the left hand side is absorbed into the Weil transformation, while in the right hand side it produces the delta-function $\delta_{A/B}(t-q)$ which allows to remove the $t$-integration. The result is given by \eqref{int-id1}.
\end{proof}
\begin{proposition}
  Let $\chi\colon A^2\to\T$ be a bi-character such that
 \begin{equation}\label{char}
\fourier{x}{y}=\chi(x,y)\chi(y,x).
\end{equation}
Then the function $\mfun\colon A^3\to\C$ defined by
\begin{equation}\label{mfun}
\mfun(x,y,z)=\chi(x,y)\wgfvm{z}{x,y}
\end{equation}
is quasi $B$-invariant in the first two arguments
\begin{multline}
\mfun(x,y+b,z)=\chi(x,b)\mfun(x,y,z),\\
\mfun(x+b,y,z)=\bar\chi(y,b)\mfun(x,y,z),\quad \forall b\in B,
\end{multline}
and satisfies the following integral relation
\begin{multline}\label{irfm}
\int_{ A/B}\mfun(p,u-t,x)\mfun(t,v-p,x+y)\mfun(v,t-q,y)
\operatorname{d}\!t\\
=\int_{ A/B}\mfun(u,s-p,y)\mfun(s,u-q,x+y)\mfun(q,v-s,x)\operatorname{d}\!s
\end{multline}
\end{proposition}
\begin{proof}
First, we verify quasi-invariance properties:
\begin{multline}
\mfun(x,y+b,z)=\chi(x,y+b)\wgfvm{z}{x,y+b}\\=\chi(x,b)\chi(x,y)\wgfvm{z}{x,y}=\chi(x,b)\mfun(x,y,z),
\end{multline}
and
 \begin{multline}
\mfun(x+b,y,z)=\chi(x+b,y)\wgfvm{z}{x+b,y}\\=\chi(b,y)\fourier{-b}{y}\mfun(x,y,z)=\bar\chi(y,b)\mfun(x,y,z).
\end{multline}
Next, we calculate the left hand side of \eqref{irfm} by substituting \eqref{mfun} and using \eqref{char}:
\begin{multline}\label{lhs}
\int_{ A/B}\mfun(p,u-t,x)\mfun(t,v-p,x+y)\mfun(v,t-q,y)
\operatorname{d}\!t
=\int_{ A/B}\chi(p,u-t)\\
\times\chi(t,v-p)\chi(v,t-q)\wgfvm{x}{p,u-t}\wgfvm{x+y}{t,v-p}\wgfvm{y}{v,t-q}
\operatorname{d}\!t\\
=\chi(p,u)\bar\chi(v,q)\int_{ A/B}\fourier{v-p}{t}
\wgfvm{x}{p,u-t}\wgfvm{x+y}{t,v-p}\wgfvm{y}{v,t-q}
\operatorname{d}\!t
\end{multline}
and, similarly, for the right hand side
\begin{multline}\label{rhs}
\int_{ A/B}\mfun(u,s-p,y)\mfun(s,u-q,x+y)\mfun(q,v-s,x)\operatorname{d}\!s=
\int_{ A/B}\chi(u,s-p)\\
\times\chi(s,u-q)\chi(q,v-s)\wgfvm{y}{u,s-p}\wgfvm{x+y}{s,u-q}\wgfvm{x}{q,v-s}\operatorname{d}\!s\\
=
\bar\chi(u,p)\chi(q,v)\int_{ A/B}\fourier{u-q}{s}\wgfvm{y}{u,s-p}\wgfvm{x+y}{s,u-q}\wgfvm{x}{q,v-s}\operatorname{d}\!s.
\end{multline}
Comparing \eqref{lhs} and \eqref{rhs}, we obtain \eqref{int-id1}. 
\end{proof}
Introducing the graphical notation
\begin{equation}
\begin{tikzpicture}[baseline=-2,scale=.7,spin/.style={circle,draw=black!50,fill=black!10,thick}]
\coordinate (ne)  at (1,1);
\coordinate (nw)  at (-1,1);
\coordinate (se)  at (1,-1);
\coordinate (sw)  at (-1,-1);
\filldraw[fill=green!5,draw=green!50!black] (-1,0)--(0,1)--(1,0)--(0,-1)--cycle;
\node (w) at (-1,0)[spin]{$a$};\node (n) at (0,1)[spin]{$b$};\node (e) at (1,0)[spin]{$c$};\node (s) at (0,-1)[spin]{$d$};
\node[color=blue] at (-1.2,-1.2){$\alpha$};\node[color=blue] at (1.2,-1.2){$\beta$};
\draw[>=stealth,->,color=blue](sw)--(ne);
\draw[>=stealth,->,color=blue](se)--(nw);
\end{tikzpicture}
=\  \mfun(a-c,d-b,\alpha-\beta)
\end{equation}
we rewrite \eqref{irfm} as a Yang--Baxter relation of IRF-type:
\begin{equation}
 \int_{A/B}\begin{tikzpicture}[baseline=-2,scale=.6,spin/.style={circle,draw=black!50,fill=black!10,thick}]
\coordinate (n)  at (90:2);
\coordinate (s)  at (-90:2);
\coordinate (sw)  at (-150:2);
\coordinate (ne)  at (30:2);
\coordinate (nw)  at (150:2);
\coordinate (se)  at (-30:2);
\filldraw[fill=green!5,draw=green!50!black] (0:2)--(60:2)--(120:2)--(180:2)--(-120:2)--(-60:2)--cycle;
\draw[draw=green!50!black] (0:2)--(0,0)--(120:2)(0,0)--(-120:2);
\node (0) at (0:2)[spin]{$a$};\node (1) at (60:2)[spin]{$b$};\node (2) at (120:2)[spin]{$c$};\node (3) at (180:2)[spin]{$d$};
\node (4) at (-120:2)[spin]{$e$};\node (5) at (-60:2)[spin]{$f$}; \node (c) at (0,0)[spin]{$g$};
\node[color=blue] at (150:2.3){$\alpha$};\node[color=blue] at (-150:2.3){$\beta$};\node[color=blue] at (-90:2.3){$\gamma$};
\draw[>=stealth,->,color=blue](sw) .. controls (160:1.3) and (80:1.3) .. (ne);
\draw[>=stealth,->,color=blue](nw) .. controls (-160:1.3) and (-80:1.3) .. (se);
\draw[>=stealth,->,color=blue](s) .. controls (-30:1.3) and (30:1.3) .. (n);
\end{tikzpicture}\ \operatorname{d}\!g
\quad =\quad 
 \int_{A/B}\begin{tikzpicture}[baseline=-2,scale=.6,spin/.style={circle,draw=black!50,fill=black!10,thick}]
\coordinate (n)  at (90:2);
\coordinate (s)  at (-90:2);
\coordinate (sw)  at (-150:2);
\coordinate (ne)  at (30:2);
\coordinate (nw)  at (150:2);
\coordinate (se)  at (-30:2);
\filldraw[fill=green!5,draw=green!50!black] (0:2)--(60:2)--(120:2)--(180:2)--(-120:2)--(-60:2)--cycle;
\draw[draw=green!50!black] (180:2)--(0,0)--(60:2)(0,0)--(-60:2);
\node (0) at (0:2)[spin]{$a$};\node (1) at (60:2)[spin]{$b$};\node (2) at (120:2)[spin]{$c$};\node (3) at (180:2)[spin]{$d$};
\node (4) at (-120:2)[spin]{$e$};\node (5) at (-60:2)[spin]{$f$}; \node (c) at (0,0)[spin]{$h$};
\node[color=blue] at (150:2.3){$\alpha$};\node[color=blue] at (-150:2.3){$\beta$};\node[color=blue] at (-90:2.3){$\gamma$};
\draw[>=stealth,->,color=blue](sw) .. controls (-100:1.3) and (-20:1.3) .. (ne);
\draw[>=stealth,->,color=blue](nw) .. controls (100:1.3) and (20:1.3) .. (se);
\draw[>=stealth,->,color=blue](s) .. controls (-150:1.3) and (150:1.3) .. (n);
\end{tikzpicture}\ \operatorname{d}\!h
\end{equation}
with the identification of the variables
\begin{multline}
x=\alpha-\beta,\quad y=\beta-\gamma,\quad p=c-e,\quad t=g-f,\\ v=c-a,\quad u=d-f,\quad s=c-h,\quad q=b-f.
\end{multline}
The essential feature  of the obtained solution is that, in general, the IRF weight function is not $B$-invariant but only quasi $B$-invariant. That means that the corresponding lattice QFT model carries a gauge symmetry given by independent $B$-translations of local spin variables. However, this gauge symmetry trivializes in the case of self-dual LCA groups of the form $A=\hat B\times B$, see below Subsection~\ref{sec:ex}.
\subsection{The Weil transform of the generalized Faddeev--Volkov weight function}
Define the Weil transform of the quantum dilogarithm
\begin{equation}
\check\phi(x,y)\equiv\int_B\phi(x+b)\fourier{y}{b}\operatorname{d}\!b,\quad \phi(x)=\int_{A/B}\check\phi(x,y)\operatorname{d}\!y.
\end{equation}
\begin{proposition}
The Weil transform of the  generalized Faddeev--Volkov weight function can be written as an integral 
of a product of Weil transforms of the quantum dilogarithm through the following formula:
 \begin{equation}
\wgfvm{x}{s,t}={1\over \phi(0)^{2}\gauss{s}\gauss{x}}
\int_{A/B}\check\phi(s-x,y+t-s-\varepsilon)\check\phi(-s-x,y)\operatorname{d}\!y.
\end{equation}
\end{proposition}
\begin{proof}
We have
\begin{multline} 
\wgfvm{x}{s,t}=\int_B\gfvm{x}{s+b}\fourier{t}{b}\operatorname{d}\!b
=\int_B{\phi(s+b-x)\phi(-s-b-x)\over\phi(0)^{2} \gauss{s+b}\gauss{x}}\fourier{t}{b}\operatorname{d}\!b\\
={1\over \phi(0)^{2}\gauss{s}\gauss{x} }\int_B\phi(s+b-x)\phi(-s-b-x)\fourier{t-s-\varepsilon}{b}\operatorname{d}\!b
\end{multline}
where, in the second equality, we used the second formula in \eqref{gfvwf}, while in the last equality 
we used the fact that, due to the property $B=B^\perp$, the restriction of the Gaussian exponential of $A$ to $B$ is a character, i.e. 
there is an element $\varepsilon\in A$ such that
\begin{equation}
\fourier{\varepsilon}{b}=\gauss{b},\quad \forall b\in B.
\end{equation}
Inserting now the inverse Weil transform for one quantum dilogarithm and, using  its quasi $B$-invariance property, we absorb the $b$-integration into the Weil transformation of the other quantum dilogarithm:
\begin{multline} 
\phi(0)^{2}\gauss{s}\gauss{x} \wgfvm{x}{s,t}\\
=\int_{A/B\times B}\phi(s+b-x)\check\phi(-s-b-x,y)\fourier{t-s-\varepsilon}{b}\operatorname{d}(y,b)\\
=\int_{A/B\times B}\phi(s+b-x)\check\phi(-s-x,y)\fourier{y+t-s-\varepsilon}{b}\operatorname{d}(y,b)\\
=\int_{A/B}\check\phi(s-x,y+t-s-\varepsilon)\check\phi(-s-x,y)\operatorname{d}\!y.
\end{multline}
\end{proof}
\section{Examples}\label{sec4}
In this section, we describe few concrete examples of quantum dilogarithms.
\subsection{Faddeev's quantum dilogarithm}  The self-dual LCA group is $A=\R$ with the Gaussian exponential and the Fourier kernel
\begin{equation}
\gauss{x}=e^{\pi\imi x^2},\quad \fourier{x}{y}=e^{2\pi\imi xy},
\end{equation}
and the quantum dilogarithm $\phi(x)=\fad{\mathsf{b}}{x}$ is described in the Introduction. The Weil transformation, associated with the subgroup $B=\Z$, in this case is also called Weil--Gelfand--Zak (WGZ) transformation:
\begin{equation}
\check\phi(x,y)=\sum_{k\in\Z}\phi(x+k)e^{2\pi\imi ky}.
\end{equation}
It has been shown in \cite{AndersenKashaev2013} that $\check\phi(x,y)$ is a meromorphic function on $\C^2$. Choosing $\chi(x,y)=e^{\pi\imi xy}$, the IRF weight function is calculated by the formula
\begin{multline}
 \mfun(x,y,z)\\=\phi(0)^{-2}e^{\pi\imi (xy-x^2-z^2)}\int_0^1\check\phi(x-z,u+y-x-1/2)\check\phi(-x-z,u)\operatorname{d}\!u
\end{multline}
which, due to periodicity in $u$, can be thought of as an integral over a closed contour in the complex plane. That means that it is a meromorphic function on $\C^3$.
\subsubsection{The special value $\mathsf{b}=e^{\pi\imi/6}$} The result of meromorphicity of $\check\phi(x,y)$ obtained in \cite{AndersenKashaev2013} can be illustrated explicitly in the special case where parameter $\mathsf{b}$ is fixed
by the value $e^{\pi\imi/6}=(\sqrt{3}+\imi)/2$, corresponding to
\begin{equation}
q_\mathsf{b}=-\bar q_\mathsf{b}=\imi e^{-\pi\sqrt{3}/2}
\end{equation}
in \eqref{qdl}. That value is such that $\check\phi(x,y)$ can be calculated explicitly by using Ramanujan's $_1\psi_1$ summation formula:
\begin{equation}\label{1psi1}
\sum_{k\in\Z}{\poc{a;q}{k}􏰉\over\poc{b;q}{k}}z^k={\poc{q;q}{\infty}\poc{{b\over a};q}{\infty}\poc{az;q}{\infty}\poc{{q\over az};q}{\infty}
\over \poc{b;q}{\infty}\poc{{q\over a};q}{\infty}\poc{z;q}{\infty}\poc{{b\over az};q}{\infty}},\quad |b/a|<|z|<1.
\end{equation}
 Indeed, assuming $\left|e^{2\pi\imi x}\right|<\left|e^{-2\pi\imi y}\right|<1$,  using the formula
\begin{equation}
\poc{xq^n;q}{\infty}={\poc{x;q}{\infty}\over\poc{x;q}{n}},
\end{equation}
and the notation
\begin{equation}
\theta_q(x)\equiv \sum_{k\in\Z}q^{k^2} x^k=\poc{q^2;q^2}{\infty}\poc{-qx;q^2}{\infty}\poc{-q/x;q^2}{\infty},
\end{equation}
we have
\begin{multline}
 \check\phi(x,y)=\sum_{k\in\Z}{\poc{-q^{1+2k}_\mathsf{b}e^{2 \pi \mathsf{b} x};q^2_\mathsf{b}}{\infty}\over
\poc{q^{1+2k}_\mathsf{b}e^{2 \pi \bar{\mathsf{b}} x};q^2_\mathsf{b}}{\infty}}e^{-2\pi\imi yk}\\
={\poc{-q_\mathsf{b}e^{2 \pi \mathsf{b} x};q^2_\mathsf{b}}{\infty}\over\poc{q_\mathsf{b}e^{2 \pi \bar{\mathsf{b}} x};q^2_\mathsf{b}}{\infty}}\sum_{k\in\Z}{\poc{q_\mathsf{b}e^{2 \pi \bar{\mathsf{b}} x};q^2_\mathsf{b}}{k}\over\poc{-q_\mathsf{b}e^{2 \pi \mathsf{b} x};q^2_\mathsf{b}}{k}
}e^{-2\pi\imi yk}\\
=
{\poc{q_\mathsf{b}^2;q_\mathsf{b}^2}{\infty}\poc{-e^{2\pi\imi x};q_\mathsf{b}^2}{\infty}\poc{q_\mathsf{b}e^{2 \pi( \bar{\mathsf{b}} x-\imi y)};q_\mathsf{b}^2}{\infty}\poc{q_\mathsf{b}e^{2 \pi(\imi y- \bar{\mathsf{b}} x)};q_\mathsf{b}^2}{\infty}
\over\poc{q_\mathsf{b}e^{2 \pi \bar{\mathsf{b}} x};q^2_\mathsf{b}}{\infty} \poc{q_\mathsf{b}e^{-2 \pi \bar{\mathsf{b}} x};q_\mathsf{b}^2}{\infty}\poc{e^{-2\pi\imi y};q_\mathsf{b}^2}{\infty}\poc{-e^{2\pi\imi (x+y)};q_\mathsf{b}^2}{\infty}}\\
=
{\poc{q_\mathsf{b}^2;q_\mathsf{b}^2}{\infty}\poc{-e^{2\pi\imi x};q_\mathsf{b}^2}{\infty}\theta_{q_\mathsf{b}}\left(-e^{2 \pi(\bar{\mathsf{b}} x-\imi y)}\right)
\over\poc{e^{-2\pi\imi y};q_\mathsf{b}^2}{\infty}\poc{-e^{2\pi\imi (x+y)};q_\mathsf{b}^2}{\infty}\theta_{q_\mathsf{b}}\left(-e^{2 \pi\bar{\mathsf{b}} x}\right)},
\end{multline}
and the IRF weight function
\begin{multline}
\mfun(x,y,z)={\poc{q_\mathsf{b}^2;q_\mathsf{b}^2}{\infty}^2\poc{-e^{2\pi\imi (x-z)};q_\mathsf{b}^2}{\infty}\poc{-e^{-2\pi\imi (x+z)};q_\mathsf{b}^2}{\infty}\over e^{\pi\imi (x^2+z^2-xy)}\phi(0)^{2} \theta_{q_\mathsf{b}}\left(-e^{2 \pi\bar{\mathsf{b}} (x-z)}\right)\theta_{q_\mathsf{b}}\left(-e^{-2 \pi\bar{\mathsf{b}} (x+z)}\right)}\\
\times\int_0^1
{\theta_{q_\mathsf{b}}\left(e^{2 \pi\imi(\imi\bar{\mathsf{b}} (z-x)+x-y-u)}\right)\theta_{q_\mathsf{b}}\left(-e^{2 \pi\imi(\imi\bar{\mathsf{b}} (x+z)-u)}\right)\operatorname{d}\!u
\over\poc{-e^{2\pi\imi (x-y-u)};q_\mathsf{b}^2}{\infty}\poc{e^{2\pi\imi (u+y-z)};q_\mathsf{b}^2}{\infty}\poc{e^{-2\pi\imi u};q_\mathsf{b}^2}{\infty}\poc{-e^{2\pi\imi (u-x-z)};q_\mathsf{b}^2}{\infty}}.
\end{multline}
\subsection{A quantum dilogarithm over $A=\R\times\Z/N\Z$} Let $N$ be a positive integer. A generalization of the previous example is the self-dual LCA group $A=\R\times\Z/N\Z$ with the Gaussian exponential and the Fourier kernel
\begin{equation}
\gauss{(x,m)}=e^{\pi\imi x^2}e^{-\pi\imi m(m+N)/N},\quad \fourier{(x,m)}{(y,n)}=e^{2\pi\imi xy}e^{-2\pi\imi mn/N}.
\end{equation}
The quantum dilogarithm found in \cite{AndersenKashaev2014} is of the form
\begin{equation}
\phi(x,m)=\prod_{j=0}^{N -1}\fad{e^{\imi\theta}}{\frac{x}{\sqrt{N }}+\imi(1-N^{-1})\cos(\theta)-\imi e^{-\imi\theta}\frac{j}{N }-\imi e^{\imi\theta}\left\{\frac{j+m}{N }\right\}}
\end{equation}
where $\theta\in]0,\pi/2[$, and $\{x\}$ denotes the fractional part of $x$. The subgroup $B$ satisfying the condition $B=B^\perp$ is isomorphic to $\Z$ and is generated by the element $(1/\sqrt{N},1)$. The bi-character, satisfying the condition~\eqref{char} exists for odd $N$, and it reads as
\begin{equation}
\chi((x,m),(y,n))=e^{\pi\imi xy}e^{\pi\imi (1-N^{-1})mn}.
\end{equation}

\subsection{Self-dual LCA groups of the form $A=\hat B\times B$}\label{sec:ex}
In that case we write
\begin{equation}
x=(\hat x,\dot x)\in \hat B\times B,
\end{equation}
 and we choose
\begin{equation}
\gauss{x}=\hat x(\dot x),\quad \fourier{x}{y}=\hat x(\dot y)\hat y(\dot x),\quad \chi(x,y)=\hat y(\dot x).
\end{equation}
The Weil transformation of the generalized Faddeev--Volkov weight function reads
\begin{multline}
\wgfvm{z}{x,y}=\int_B\gfvm{z}{(\hat x,\dot x+b)}\hat y(b)\operatorname{d}\!b=\int_B\gfvm{z}{(\hat x,b)}\hat y(b-\dot x)\operatorname{d}\!b\\
=\bar\chi(x,y)\int_B\gfvm{z}{(\hat x,b)}\hat y(b)\operatorname{d}\!b,
\end{multline}
so that the IRF weight function, given by the formula
\begin{equation}\label{irfb^b}
\mfun(x,y,z)=\int_B\gfvm{z}{(\hat x,b)}\hat y(b)\operatorname{d}\!b
=\hat x(\dot z)\int_B{\phi((\hat x-\hat z,b-\dot z))\over\phi((\hat x+\hat z,b+\dot z))}\hat z(b)\hat y(b)\operatorname{d}\!b,
\end{equation}
is $B$-invariant in the first two arguments and thus there is no non-trivial gauge symmetry in the corresponding IRF lattice QFT-model. Notice that in this weight function the only variable $\dot z$ which takes its values in the subgroup $B$ enters as a part of the spectral parameter.
\subsubsection{The tropical quantum dilogarithm over $A=\T\times \Z $}
 Let us choose $A=\T\times \Z $ with
 \begin{equation}\label{gexptz}
\gauss{(z,m)}=z^m,\quad
\fourier{(z,m)}{(w,n)}=z^nw^m.
\end{equation}
The simplest known quantum dilogarithm in that case is given by the following `tropical' formula
\begin{equation}\label{tqdl}
\phi(z,m)=z^{(m)_+},\quad (m)_+\equiv\max(m,0),\quad \forall (z,m)\in\T\times\Z.
\end{equation}
For a function $f\in\C^{\T\times\Z}$, the Weil transformation associated to the subgroup $B=\Z$ reduces to taking the Fourier series over the second argument:
\begin{multline}
\check f((u,m),(v,n))=\sum_{k\in\Z}f(z,m+k)\fourier{v,n}{1,k}=\sum_{k\in\Z}f(u,m+k)v^k\\
 =v^{-m}\sum_{k\in\Z}f(u,k)v^k.
\end{multline}
In the case of the tropical quantum dilogarithm~\eqref{tqdl}, assuming that $|u|<|v|^{-1}<1$, we have
\begin{multline}
 v^m\check\phi((u,m),(v,n))=\sum_{k\in\Z}\phi(u,k)v^k=\sum_{k<0}v^k+\sum_{k\ge0}(uv)^k\\
 =\sum_{k\ge0}v^{-1-k}+\sum_{k\ge0}(uv)^k ={1\over v-1}+{1\over 1- uv}
 ={(u-1)v\over (v-1)(uv-1)}.
\end{multline}
Choosing $\chi((x,k),(y,l))=y^k$ and assuming that $\max(|x|,|y|)<r<\min(|z|,|xyz|)$,  we write for the IRF weight function  
\begin{multline}
\mfun((x,k),(y,l),(z,m))=y^k \wgfvm{(z,m)}{(x,k),(y,l)}\\
 ={y^k\over x^kz^m}\oint_{\T r/|y|}\check\phi((x/z,k-m),(uy/x,0))\check\phi((1/xz,-k-m),(u,0)){\operatorname{d}\!u\over 2\pi\imi u}\\
 ={y^k\over x^kz^m}\oint_{\T r/|y|}{(x/z-1)(uy/ x)^{1-k+m}\over(uy/x-1)(uy/z-1)}
 {(1/xz-1)u^{1+k+m}\over(u-1)(u/xz-1)}{\operatorname{d}\!u\over 2\pi\imi u}\\
  ={(x-z)(1-xz)\over 2\pi\imi}\oint_{\T r/|y|}{uy(u^2y/xz)^{m}\operatorname{d}\!u\over(uy-x)(uy-z)(u-1)(u-xz)}\\
  ={(x-z)y(1-xz)\over 2\pi\imi(xyz)^m}\oint_{\T r}{v^{2m+1}\operatorname{d}\!v\over (v-x)(v-y)(v-z)(v-xyz)}.
\end{multline}
The contour integral here can be evaluated easily by the method of residues, but the final formula is less compact.
The same weight function, by using formula~\eqref{irfb^b}, can also be written as a single sum over integers
\begin{equation}
\mfun((x,k),(y,l),(z,m))=x^{m}\sum_{n\in\Z}{\left(x/z\right)}^{(n-m)_+}(xz)^{-(n+m)_+}(zy)^n.
\end{equation}
\subsubsection{The DGG quantum dilogarithm over $\T\times\Z$} Taking the same Gaussian exponential as in \eqref{gexptz}, a  quantum dilogarithm over $\T\times\Z$ can be extracted from  the tetrahedron index of Dimofte--Gaiotto--Gukov \cite{MR3262519}. The explicit formula is as follows:
\begin{equation}
\phi(z,m)={\poc{-q^{1-m}z;q^2}{\infty}\over\poc{-q^{1-m} / z;q^2}{\infty}}
\end{equation}
where $q\in]-1,1[$ is a fixed parameter. The Weil transform of this function can be calculated explicitly by using Ramanujan's ${_1\psi_1}$ summation formula~\eqref{1psi1}.
Indeed, assuming that $|u|<|v|^{-1}<1$, we calculate
\begin{multline}
 v^m\check\phi((u,m),(v,n))=\sum_{k\in\Z}\phi_q(u,k)v^k=\sum_{k\in\Z}{\poc{-q^{1+k}u;q^2}{\infty}\over\poc{-q^{1+k}/u;q^2}{\infty}}v^{-k}\\
 =\sum_{k\in\Z}{\poc{-q^{1+2k}u;q^2}{\infty}\over\poc{-q^{1+2k}/u;q^2}{\infty}}v^{-2k}
 +\sum_{k\in\Z}{\poc{-q^{2k}u;q^2}{\infty}\over\poc{-q^{2k}/u;q^2}{\infty}}v^{1-2k}\\
 ={\poc{-qu;q^2}{\infty}\over\poc{-q/u;q^2}{\infty}}\sum_{k\in\Z}{\poc{-q/u;q^2}{k}\over\poc{-qu;q^2}{k}}v^{-2k}
 +{\poc{-u;q^2}{\infty}\over \poc{-1/u;q^2}{\infty}}\sum_{k\in\Z}{\poc{-1/u;q^2}{k}\over\poc{-u;q^2}{k}}v^{1-2k}\\
 ={\poc{q^2;q^2}{\infty}\poc{u^2;q^2}{\infty}\poc{-{q\over uv^2};q^2}{\infty}\poc{-quv^2;q^2}{\infty}
\over\poc{-q/u;q^2}{\infty}\poc{-qu;q^2}{\infty}\poc{1/v^2;q^2}{\infty}\poc{u^2v^2;q^2}{\infty}}\\
 +{v\poc{q^2;q^2}{\infty}\poc{u^2;q^2}{\infty}\poc{-1/uv^2;q^2}{\infty}\poc{-q^2uv^2;q^2}{\infty}
\over\poc{-1/u;q^2}{\infty} \poc{-q^2u;q^2}{\infty}\poc{1/v^2;q^2}{\infty}\poc{u^2v^2;q^2}{\infty}}\\
={\poc{q^2;q^2}{\infty}\poc{u^2;q^2}{\infty}\over\poc{1/v^2;q^2}{\infty}\poc{u^2v^2;q^2}{\infty}}\left({\theta_q(uv^2)
\over\theta_q(u)}+v{\theta_q(quv^2)
\over\theta_q(qu)}\right).
\end{multline}
Thus, assuming  $\max(|x|,|y|)<r<\min(|z|,|xyz|)$, the IRF weight function reads
\begin{multline}
 \mfun((x,k),(y,l),(z,m))=y^k \wgfvm{(z,m)}{(x,k),(y,l)}\\
 ={y^k\over x^kz^m}\oint_{\T r/|y|}\check\phi((x/z,k-m),(uy/x,0))\check\phi((1/xz,-k-m),(u,0)){\operatorname{d}\!u\over 2\pi\imi u}\\
 ={y^k\over x^kz^m}\oint_{\T r/|y|}
 {\poc{q^2;q^2}{\infty}\poc{x^2/ z^2;q^2}{\infty}
  \left({\theta_q(y^2u^2/xz)
\over\theta_q(x/ z)}+{yu\theta_q(qy^2u^2/xz)
\over x\theta_q(qx/z)}\right)
 \over\left(yu/x\right)^{k-m}\poc{x^2/y^2u^2;q^2}{\infty}\poc{y^2u^2/z^2;q^2}{\infty}}
\\
 \times u^{k+m}
{\poc{q^2;q^2}{\infty}\poc{1/x^2z^2;q^2}{\infty}\over\poc{1/u^2;q^2}{\infty}\poc{u^2/x^2z^2;q^2}{\infty}}\left({\theta_q(u^2/xz)
\over\theta_q(1/xz)}+u{\theta_q(qu^2/xz)
\over\theta_q(q/xz)}\right)
 {\operatorname{d}\!u\over 2\pi\imi u}\\
  ={\poc{q^2;q^2}{\infty}^2\poc{x^2/z^2;q^2}{\infty}\poc{1/x^2z^2;q^2}{\infty}\over 2\pi\imi (xz/y)^m}\\
 \times\oint_{\T r/|y|}{\left({\theta_q(y^2u^2/xz)
\over\theta_q(x/ z)}+{yu\theta_q(qy^2u^2/xz)
\over x\theta_q(qx/z)}\right)\left({\theta_q(u^2/xz)
\over\theta_q(1/xz)}+u{\theta_q(qu^2/xz)
\over\theta_q(q/xz)}\right)u^{2m-1}\operatorname{d}\!u\over\poc{x^2/y^2u^2;q^2}{\infty}\poc{u^2y^2/z^2;q^2}{\infty}\poc{1/u^2;q^2}{\infty}\poc{u^2/x^2z^2;q^2}{\infty}}\\
 ={\poc{q^2;q^2}{\infty}^2\poc{x^2/z^2;q^2}{\infty}\poc{1/x^2z^2;q^2}{\infty}\over 2\pi\imi (xyz)^m}\\
 \times\oint_{\T r}{\left({\theta_q(v^2/xz)
\over\theta_q(x/z)}+{v\theta_q(qv^2/xz)
\over x\theta_q(qx/z)}\right)\left({\theta_q(v^2/xy^2z)
\over\theta_q(1/xz)}+{v\theta_q(qv^2/xy^2z)
\over y\theta_q(q/xz)}\right)v^{2m-1}\operatorname{d}\!v\over \poc{x^2/v^2;q^2}{\infty}\poc{v^2/z^2;q^2}{\infty}\poc{y^2/v^2;q^2}{\infty}\poc{v^2/x^2y^2z^2;q^2}{\infty}}.
\end{multline}
\def\cprime{$'$} \def\cprime{$'$}


\begin{thebibliography}{10}

\bibitem{AndersenKashaev2013}
J{\o}rgen~Ellegaard Andersen and Rinat Kashaev.
\newblock A new formulation of the {T}eichm\"uller {TQFT}.
\newblock arXiv:1305.4291, 2013.

\bibitem{AndersenKashaev2014}
J{\o}rgen~Ellegaard Andersen and Rinat Kashaev.
\newblock Complex quantum {C}hern--{S}imons.
\newblock arXiv:1409.1208, 2014.

\bibitem{MR2353414}
Vladimir~V. Bazhanov, Vladimir~V. Mangazeev, and Sergey~M. Sergeev.
\newblock Faddeev-{V}olkov solution of the {Y}ang-{B}axter equation and
  discrete conformal symmetry.
\newblock {\em Nuclear Phys. B}, 784(3):234--258, 2007.

\bibitem{MR2396233}
Vladimir~V. Bazhanov, Vladimir~V. Mangazeev, and Sergey~M. Sergeev.
\newblock Exact solution of the {F}addeev-{V}olkov model.
\newblock {\em Phys. Lett. A}, 372(10):1547--1550, 2008.

\bibitem{MR2730793}
Vladimir~V. Bazhanov, Vladimir~V. Mangazeev, and Sergey~M. Sergeev.
\newblock Quantum geometry of 3-dimensional lattices and tetrahedron equation.
\newblock In {\em X{VI}th {I}nternational {C}ongress on {M}athematical
  {P}hysics}, pages 23--44. World Sci. Publ., Hackensack, NJ, 2010.

\bibitem{MR2868616}
Vladimir~V. Bazhanov and Sergey~M. Sergeev.
\newblock Elliptic gamma-function and multi-spin solutions of the
  {Y}ang-{B}axter equation.
\newblock {\em Nuclear Phys. B}, 856(2):475--496, 2012.

\bibitem{MR3019403}
Vladimir~V. Bazhanov and Sergey~M. Sergeev.
\newblock A master solution of the quantum {Y}ang-{B}axter equation and
  classical discrete integrable equations.
\newblock {\em Adv. Theor. Math. Phys.}, 16(1):65--95, 2012.

\bibitem{MR3262519}
Tudor Dimofte, Davide Gaiotto, and Sergei Gukov.
\newblock 3-manifolds and 3d indices.
\newblock {\em Adv. Theor. Math. Phys.}, 17(5):975--1076, 2013.

\bibitem{Faddeev1994}
L.~D. Faddeev.
\newblock Current-like variables in massive and massless integrable models.
\newblock arXiv:hep-th/9408041, 1994.

\bibitem{MR1345554}
L.~D. Faddeev.
\newblock Discrete {H}eisenberg--{W}eyl group and modular group.
\newblock {\em Lett. Math. Phys.}, 34(3):249--254, 1995.

\bibitem{MR1264393}
L.~D. Faddeev and R.~M. Kashaev.
\newblock Quantum dilogarithm.
\newblock {\em Modern Phys. Lett. A}, 9(5):427--434, 1994.

\bibitem{MR0039154}
I.~M. Gel{\cprime}fand.
\newblock Expansion in characteristic functions of an equation with periodic
  coefficients.
\newblock {\em Doklady Akad. Nauk SSSR (N.S.)}, 73:1117--1120, 1950.

\bibitem{MR1865275}
Kazuhiro Hikami.
\newblock Hyperbolicity of partition function and quantum gravity.
\newblock {\em Nuclear Phys. B}, 616(3):537--548, 2001.

\bibitem{MR1607296}
R.~M. Kashaev.
\newblock Quantization of {T}eichm\"uller spaces and the quantum dilogarithm.
\newblock {\em Lett. Math. Phys.}, 43(2):105--115, 1998.

\bibitem{KashaevMarino2015}
Rinat Kashaev and Marcos Mari\~no.
\newblock Operators from mirror curves and the quantum dilogarithm.
\newblock arXiv:1501.01014, 2015.

\bibitem{MR2861174}
Rinat~M. Kashaev and Tomoki Nakanishi.
\newblock Classical and quantum dilogarithm identities.
\newblock {\em SIGMA Symmetry Integrability Geom. Methods Appl.}, 7:Paper 102,
  29, 2011.

\bibitem{MR3190461}
Vladimir~V. Mangazeev.
\newblock On the {Y}ang-{B}axter equation for the six-vertex model.
\newblock {\em Nuclear Phys. B}, 882:70--96, 2014.

\bibitem{MR3246845}
Vladimir~V. Mangazeev.
\newblock {$Q$}-operators in the six-vertex model.
\newblock {\em Nuclear Phys. B}, 886:166--184, 2014.

\bibitem{MR3138832}
Vladimir~V. Mangazeev, Vladimir~V. Bazhanov, and Sergey~M. Sergeev.
\newblock An integrable 3{D} lattice model with positive {B}oltzmann weights.
\newblock {\em J. Phys. A}, 46(46):465206, 16, 2013.

\bibitem{MR1846786}
V.~P. Spiridonov.
\newblock On the elliptic beta function.
\newblock {\em Uspekhi Mat. Nauk}, 56(1(337)):181--182, 2001.

\bibitem{MR0005741}
Andr{\'e} Weil.
\newblock {\em L'int\'egration dans les groupes topologiques et ses
  applications}.
\newblock Actual. Sci. Ind., no. 869. Hermann et Cie., Paris, 1940.
\newblock [This book has been republished by the author at Princeton, N. J.,
  1941.].

\bibitem{MR1770545}
S.~L. Woronowicz.
\newblock Quantum exponential function.
\newblock {\em Rev. Math. Phys.}, 12(6):873--920, 2000.

\bibitem{Zak1967}
J.~Zak.
\newblock Finite translations in solid state physics.
\newblock {\em Phys. Rev. Lett.}, 19:1385--1397, 1967.

\end{thebibliography}
\end{document}